\documentclass{amsart}

\newtheorem{thm}{Theorem}[section]
\newtheorem{lemma}{Lemma}[section]
\newtheorem{prop}[lemma]{Proposition}
\newtheorem{cor}[lemma]{Corollary}

\theoremstyle{remark}
\newtheorem{remark}[thm]{Remark}
\numberwithin{equation}{section}

\begin{document}
\title[$N$-body time-symmetric initial data sets in general relativity]{Construction of $N$-body time-symmetric initial data sets in general relativity}
\author{Piotr T. Chru\'{s}ciel} 
\address{Hertford College and Oxford Centre for Nonlinear PDE, University of Oxford}
\email{chrusciel@maths.ox.ac.uk}
\author{Justin Corvino} 
\address{Department of Mathematics, Lafayette College}
\email{corvinoj@lafayette.edu}
\author{James Isenberg}
\address{Department of Mathematics, University of Oregon}
\email{isenberg@uoregon.edu}
\thanks{The authors are grateful to Institut Mittag-Leffler (Djursholm, Sweden), for hospitality and financial support during the initial work on this paper.  JC acknowledges partial support by both the Fulbright Foundation, and from NSF grant DMS-0707317.  JI is partially supported by NSF grant PHY-0652903.  JC and JI thank Daniel Pollack for several useful conversations.}


\begin{abstract} Given a collection of $N$ asymptotically Euclidean ends with zero scalar curvature, we construct a Riemannian manifold with zero scalar curvature and one asymptotically Euclidean end, whose boundary has a neighborhood isometric to the disjoint union of a specified collection of sub-regions of the given ends.   An application is the construction of time-symmetric solutions of the constraint equations which model $N$-body initial data.  

\end{abstract}

\maketitle

\section{Introduction.}  

The work of Corvino and Schoen \cite{cor:schw, cs:ak} shows that one can glue an interior region of any asymptotically Euclidean (AE) solution of the vacuum constraints to an exterior Kerr or Schwarzschild solution in the following sense: for any chosen AE end in any given AE solution of the vacuum Einstein constraint equations, and for any chosen interior region (bounded away from infinity) in this solution, there is a new AE solution of the vacuum constraints which has an interior region that is isometric to the original interior region, and so that the chosen AE end has been replaced with an AE end in a space-like slice of a Kerr space-time.  If the original solution of the constraints is time-symmetric, then the exterior is isometric to an AE end in a time-symmetric slice of a Schwarzschild space-time.  This work shows, remarkably, that one can geometrically shield the details of the gravitational field in the interior region, so that sufficiently far away, all that one sees is the effective total energy-momentum and total angular momentum of the interior.

A natural question to pose is whether one can produce solutions of the constraints which contain exact copies of the interior regions of two or more solutions. Using the gluing techniques of \cite{cd:as, cd,imp, IMP2}, one can indeed do this \cite{cip}; however the solutions resulting from these gluing methods contain two or more asymptotic ends. There are gluing results \cite{cd} which do incorporate multiple interiors and can produce solutions with a single asymptotic end, but all such results to date require that the multiple regions be copies of each other, and be placed in certain symmetric configurations.

We show here that in fact one can glue multiple interior regions of solutions of the vacuum constraint equations into a solution with a single asymptotic (Schwarzschild) end, without imposing any matching conditions on the solutions being glued, and without imposing any symmetry restrictions on the placement of the various interior regions. Our main restriction here is that we consider only time-symmetric solutions of the constraint equations (and therefore Schwarzschild rather than Kerr exteriors), and in the solutions constructed, the given interior regions must be sufficiently far from each other. In future work, we lift the time-symmetric restriction.  It is not likely that we can weaken the condition that the regions be sufficiently distant from each other, or get a more effective bound on the scale of the final configuration. 

The motivation for the present work is to set up initial data for the general relativistic (Einsteinian) version of the gravitational $N$-body problem. We recall that one of the signature results of eighteenth century physics is the complete and explicit solution of the Newtonian two-body problem for point particles. There is no such complete solution for more than two bodies, or for bodies with finite extent and interior dynamics, like stars. However, for any number of bodies, and for most standard equations of state, choosing $N$-body initial data for Newtonian gravity involves solving a single (linear) Poisson equation. While the nonlinearity and intricacy of the Einstein initial value constraint equations \cite {bi:ce} has long suggested that choosing $N$-body initial data for Einsteinian gravity would be much harder, this work (and its followup) indicate that in fact one can set up fairly general $N$-body initial data in general relativity.

As noted above, in this paper our focus is on time-symmetric solutions of the vacuum constraints, in which case the gravitational field is described completely by a Riemannian metric $g$, which satisfies the single equation $R(g)=0$, where $R(g)$ is the scalar curvature of $g$. The aim is to start with a collection of $N$ asymptotically Euclidean time-symmetric solutions, and choose a fixed interior region in each, thereby specifying each of the ``bodies." Then after choosing a collection of points which roughly locate the $N$ bodies on a fiducial flat background, one shows that there is a solution of the constraints which: i) contains $N$ regions which are isometric to the chosen bodies; ii) is a solution of the time-symmetric constraint $R(g)=0$ everywhere; iii) is identical to a slice of Schwarzschild sufficiently far from the bodies; and iv) has the centers of the bodies in a configuration which is a
scaled version of the chosen configuration.  Here the scale factor can be chosen arbitrarily above a certain threshold. For example, if $N=2$, the distance between the bodies can be arbitrarily chosen above a certain threshold value (depending on the ratio of the masses of the bodies).

To carry out such a construction, we start by recalling that a Riemannian metric $g$ on the exterior of a ball in $\mathbb{R}^n$ is an \emph{asymptotically Euclidean end} (to order $\ell$) provided there are coordinates in which, for multi-indices $|\beta|\leq \ell$, 
\begin{eqnarray} \label{af}
|\partial_x^{\beta}(g_{ij}-\delta_{ij})(x)|&=&O\big(|x|^{-|\beta|-(n-2)}\big)\; ,
\end{eqnarray}
where $\partial$ is a partial derivative operator.  We require $\ell \geq 3$ in what follows, and we take $g$ to be smooth, at least $C^{4, \alpha}$ for some $\alpha \in (0,1)$.  A complete Riemannian manifold $(M,g)$ is \emph{asymptotically Euclidean} if there is a compact set $K \subset M$ so that $M\setminus K$ is the disjoint union of (finitely many) AE ends.  

We now state the main theorem.  

\begin{thm}  For each $k=1, \ldots, N$, let $(E_k, g^k)$ be an $n$-dimensional ($n\geq 3$ AE end with zero scalar curvature and ADM mass $m_k>0$, and let $U_k\subset E_k$ be a pre-compact neighborhood of the boundary $\partial E_k$.  Let $M_{\rm{ext}}= \mathbb{R}^n\setminus \bigcup\limits_{k=1}^N B_k$, where $\overline{B}_1, \ldots, \overline{B}_N$ are pairwise disjoint closed balls in $\mathbb{R}^n$.  Then for each $\epsilon>0$, there is a Riemannian metric $g_{\epsilon}$ on $M_{\rm{ext}}$ of zero scalar curvature, so that $(M_{\rm{ext}},g_{\epsilon})$ contains a neighborhood of the boundary $\partial M_{\rm{ext}}$ which is isometric to the disjoint union $\bigcup\limits_{k=1}^N (U_k, g^k)$, and $(M_{\rm{ext}}, g_{\epsilon})$ has one AE end which is identical to a Schwarzschild exterior with mass $m(g_{\epsilon})$ satisfying $\Big| m(g_{\epsilon})-\sum\limits_{i=1}^k m_k\Big| < \epsilon$. \label{main} \end{thm}

The theorem allows us to glue together any finite number of AE ends which solve the time-symmetric vacuum constraint equation (i.e., which have vanishing scalar curvature) into a solution with \emph{one} AE end.  Moreover, the construction is local near infinity in each end $E_k$, i.e., any given compact subset of $E_k$ can be realized isometrically in the final metric $(M,g_{\epsilon})$.  As a direct consequence, then, we can take a collection of $N$ Riemannian $n$-manifolds $(M_k, g^k)$ in each of which we have chosen an AE end $E_k$ of zero scalar curvature and positive ADM mass $m_k$, we let $M_{\text{ext}}= \mathbb{R}^n\setminus \bigcup\limits_{k=1}^N B_k$ as above, and we can produce a Riemannian manifold $(M, g_{\epsilon})$ with the following properties: $M$ is diffeomorphic to the quotient of the disjoint union $M_{\text{ext}} \bigcup\left(\bigcup\limits_{k=1}^N M_k\setminus E_k\right)$, where we identify $\partial E_k$ and $\partial B_k$; $g_{\epsilon}=g^k$ on $M_k\setminus E_k$; $(M_{\text{ext}}, g_{\epsilon})$ has vanishing scalar curvature and one AE end, with ADM mass $m(g_{\epsilon})$ satisfying $\Big| m(g_{\epsilon})-\sum\limits_{i=1}^k m_k\Big| < \epsilon$.  The result of this kind of construction may be interpreted as an $N$-body initial data set.  There is some flexibility from the fact that the construction only deforms the original solutions near infinity in the chosen ends.  For instance, we can allow the $N$ original solutions (bodies) to have multiple ends, and we can allow nonzero scalar curvature supported away from the chosen ends, which might be interpreted as energy density for matter or a field inside the bodies.   As one of the goals is to glue multiple AE ends into a single AE end, we state the following corollary, which follows directly from Theorem \ref{main}; note that if each of the $M_k$ below is diffeomorphic to $\mathbb{R}^n$, then from the above discussion, we can choose $M$ to be diffeomorphic to $\mathbb{R}^n$ as well. 

\begin{cor} Let $(M_k, g^k)$, $n\geq 3$, $k=1, \ldots, N$, be complete AE solutions of the time-symmetric constraint $R(g^k)=0$, each of which has one AE end (\ref{af}), with ADM mass $m_k>0$.  Let $U_k\subset M_k$ be chosen compactly contained sub-domains.  Then for each $\epsilon>0$, there is a complete Riemannian manifold $(M,g_{\epsilon})$ with $R(g_{\epsilon})=0$, which contains a region isometric to the disjoint union $\bigcup\limits_{k=1}^N (U_k, g^k)$, and with only one AE end which is identical to a Schwarzschild end with mass $m(g_{\epsilon})$ satisfying $\Big| m(g_{\epsilon})-\sum\limits_{i=1}^k m_k\Big| < \epsilon$.  \label{maincor} \end{cor}

\section{Preliminaries}

We begin with a proposition which follows from \cite{cd:as} or \cite{cor:as}.   The metrics introduced here give a family of metrics into which the ends will be glued to prove the main theorem. 

\begin{prop}  There is a family of metrics $\gamma_{t}$, $t\in (0,\delta)$, on $\mathbb{R}^n$ of zero scalar curvature which converge as $t\rightarrow 0^+$ to the Euclidean metric in $C^k$ (for any $k$), and are exactly Schwarzschild on $\{ x: |x|\geq 1\}$ of positive mass $m(\gamma_t)$ and centered at $x=0$, with $\lim\limits_{t\rightarrow 0^+} m(\gamma_t)=0$. \label{prop:as}
\end{prop}

We also recall two results from \cite{cor:schw}, cf. \cite{cd, cs:ak}.  Let $K$ be the span of the set of linear and constant functions $K:=\mbox{span}\{ 1, x^1, \ldots, x^n\}$.

\begin{prop} Let $\Omega\subset \mathbb{R}^n$ be a bounded smooth domain, and let $\zeta\in C_c^{\infty}(\Omega)$ be a cutoff function.  There is a neighborhood $\mathcal U$ of the flat metric in $C^{4,\alpha}(\overline{\Omega})$ so that for any $\Omega_0\subset\subset \Omega$, there is an $\epsilon_0 >0$ so that for smooth metrics $g\in \mathcal U$ and functions $S\in C_c^{\infty}(\Omega_0)$ with $\|S\|_{C^{0,\alpha}}< \epsilon_0$, there is a smooth metric $g+h$ so that $h$ is supported in $\overline{\Omega}$ and satisfies a bound $\| h\|_{C^{2,\alpha}}\leq C \|S\|_{C^{0,\alpha}}$, and with 
\[R(g+h)-(R(g)+S) \in \zeta K.\] The map $(g,S)\mapsto h$ is continuous. \label{prop:proj} \end{prop}

We will see in the proof of Theorem \ref{main} in the next section how this proposition is a key ingredient in the proof of the following theorem. 

\begin{thm} Let $(M,g)$ be an $n$-manifold with zero scalar curvature containing an AE end $\mathcal E$ of ADM mass $m(g) \neq 0$, and let $\mathcal E_r\subset \mathcal E$ be the region corresponding to $\{ x: |x| > r\}$ in AE coordinates.  Let $k$ be a nonnegative integer.  Then for any $\epsilon>0$, there is an $R>0$ and a (smooth) metric $\bar{g}$ with zero scalar curvature and $\|g-\bar{g}\|_{C^k(\mathcal E)}<\epsilon$ (norm measured with respect to the Euclidean metric in the AE coordinate chart), with $|m(g)-m(\bar{g})|< \epsilon$, and so that $\bar{g}$ is equal to $g$ on $M\setminus \mathcal E_{R}$, and $\bar{g}$ is identical to an AE end of a standard Schwarzschild slice on $\mathcal E_{2R}$.  \label{thm:schw}
\end{thm}

\begin{remark} This statement is a slight generalization of the result in \cite{cor:schw}, in which $g$ had the following form on $\mathcal E$: $g_{ij}(x)=\big(1+\frac{m}{(n-1)|x|^{n-2}}\big)^{4/(n-2)} \delta_{ij}+h_{ij}(x)$, with $\left| \partial_x^{\beta} h_{ij}(x)\right|=O\big(|x|^{-|\beta|-(n-1)}\big)$ for $|\beta|\leq 3$.  Under this condition, or more generally assuming that $g$ has enough approximate parity symmetry so that the center of mass is well-defined \cite{cd, cs:ak}, then the center of the Schwarzschild end can be chosen close to that of $(\mathcal E, g)$.  The proof of Theorem \ref{thm:schw} is a trivial modification of the original argument.  \end{remark}

\section{Proof of the main theorem}

We now use the results assembled in the preceding section to prove the main theorem.  We will prove the result in dimension $n=3$; the argument in the other dimensions $n>3$ is essentially identical. 

\begin{proof}[Proof of Theorem \ref{main}]   Choose $c_1, \ldots, c_N\in \mathbb{R}^3$, with $\sum\limits_{k=1}^N m_k c_k =0$ and with $|c_k|>5$.  Moreover we arrange the closed balls $\{x: |x-c_k|\leq 4\}$ to be pairwise disjoint.  Let $B_0\subset \mathbb{R}^3$ be a Euclidean ball of radius $R_0=5+\max \{ |c_1|, \ldots, |c_N|\}$, centered at the origin. Let $m_T=\sum\limits_{k=1}^N m_k$.

We now construct a family of metrics $\tilde g_{(\epsilon,\mu,c)}$ on $B_0\setminus \Big( \bigcup\limits_{k=1}^N \{ x: |x-c_k|\leq 1\}\Big)$, for $0<\epsilon \ll 1$, $1/2 < \mu< 3/2$, and $c\in \mathbb{R}^3$ with $|c|<1$.  Let $\bar\gamma_{\epsilon}=\gamma_{t_{\epsilon}}$ be one of the metrics from the Proposition \ref{prop:as}, with mass $m(\bar{\gamma}_{\epsilon})=\epsilon m_{T}$, and let $\psi$ be a smooth nondecreasing function so that $\psi(t)= 0$ for $t< 9/4$ and $\psi(t)=1$ for $t>11/4$.  Define $\tilde g_{(\epsilon,\mu,c)}$ on $B_0\setminus \Big( \bigcup\limits_{k=1}^N \{ x: |x-c_k|\leq 1\}\Big)$ as follows:  

\begin{enumerate}
\item[$\bullet$] on each $A'_k=\{ x: 1<|x-c_k|<2\}$, $\tilde g_{(\epsilon,\mu,c)}(x)=\big(1+\frac{\epsilon m_k}{2|x-c_k|}\big)^4 g_{\text{Eucl}}(x)$
\item[$\bullet$] on each $A_k=\{ x: 2\leq |x-c_k|\leq 3\}$,\\ $\tilde{g}_{(\epsilon,\mu,c)}(x)=(1-\psi(|x-c_k|))\big(1+\frac{\epsilon m_k}{2|x-c_k|}\big)^4 g_{\text{Eucl}}(x)+\psi(|x-c_k|) \bar{\gamma}_{\epsilon}(\mu^{-1}(x-c))$
\item[$\bullet$] on $B_0\setminus \Big( \bigcup\limits_{k=1}^N \{x: |x-c_k|<3\}\Big)$, $\tilde g_{(\epsilon,\mu,c)}(x)= \bar{\gamma}_{\epsilon}(\mu^{-1}(x-c))$.
\end{enumerate}

Let $\Omega=B_0 \setminus \Big( \bigcup\limits_{k=1}^N \{ x: |x-c_k|\leq 2 \}\Big)$.  Let $\zeta$ be a smooth cutoff function, which we can define as follows: 

\begin{enumerate}
\item[$\bullet$]  $\zeta(x)=1$ on $ \{ x: |x|\leq R_0-1\} \setminus \Big( \bigcup\limits_{k=1}^N \{ x: |x-c_k|\leq 3 \}\Big)$\\
\item[$\bullet$]  $\zeta(x)= \psi(|x-c_k|)$ on $A_k=\{ x: 2\leq |x-c_k|\leq 3\}$\\
\item[$\bullet$] $\zeta(x)= \psi(R_0+2-|x|)$ on $\{ x: R_0-1\leq |x|\leq R_0\}$. 
\end{enumerate}

\medskip

There is an $\epsilon_1>0$ so that for all $(\epsilon, \mu, c)\in \Theta=\{ (\epsilon, \mu, c)\in \mathbb{R}\times\mathbb{R}\times \mathbb{R}^3: 0<\epsilon<\epsilon_1, \; 1/2< \mu<3/2,\; |c|<1\}$, $\tilde{g}_{(\epsilon, \mu,c)}$ is in the neighborhood $\mathcal U$ from Proposition \ref{prop:proj}, $\|\tilde{g}_{(\epsilon, \mu,c)}-g_{\text{Eucl}}\|_{C^5}\leq C \epsilon$, and moreover $\|R(\tilde{g}_{(\epsilon,\mu,c)})\|_{C^{0,\alpha}}<\epsilon_0$.  Thus there is a constant $C>0$ so that for any $\theta:=(\epsilon, \mu, c)\in \Theta$, there is a small deformation $h^{\theta}$ (obtained by applying Proposition \ref{prop:proj} with $g=\tilde{g}_{(\epsilon, \mu, c)}$ and $S=-R(\tilde{g}_{(\epsilon,\mu,c)})$) with the following properties: $h^{\theta}$ is smooth on $B_0$ and supported on $\overline{\Omega}$, $R(\tilde{g}_{(\epsilon,\mu,c)}+h^{\theta})= \sum\limits_{j=0}^3 a_j \zeta x^j$, where $x^0:=1$, and $\|h^{\theta}\|_{C^3(\Omega)}+\sum\limits_{j=0}^3 |a_j|\leq C\epsilon$.  Furthermore, given $\epsilon$, each $a_j$ is continuous in $(\mu, c)$.  We argue that for $\epsilon$ small enough, by choosing $\mu$ and $c$ appropriately, we can arrange for the integrals $\int_{\Omega} x^\ell R(\tilde{g}_{(\epsilon,\mu,c)}+h^{\theta})\; dx$ to vanish for
$\ell=0, 1, 2, 3$.  From this it follows immediately that $a_j=0$ for $j=0, 1, 2, 3$, and hence $R(\tilde{g}_{(\epsilon,\mu,c)}+h^{\theta})= 0$, as desired.  

 Let $\tilde{g}_{(\epsilon, \mu, c)}=\tilde{g}^{\theta}$ for notational convenience, and let $\tilde{h}^{\theta}=\tilde g^{\theta}-g_{\text{Eucl}}$.   Note that we can choose $C>0$ so that $\|\tilde{h}^{\theta}\|_{C^3(\Omega)}\leq C\epsilon$ for all $\theta \in \Theta$.  

Now let $L(h)=-\Delta(tr(h))+div(div(h))= \sum\limits_{i,j=1}^3 (h_{ij,ij}-h_{ii,jj})$ be the linearization at the Euclidean metric of the scalar curvature operator $g\mapsto R(g)$.  Note that $L^*(x^\ell)=0$ for $\ell =0, 1, 2, 3$.  We also note the coordinate formula for the scalar curvature 
\begin{eqnarray}
R(g)&=&\sum\limits_{i,j,k=1}^3g^{ij}\left( \Gamma^k_{ij,k}-\Gamma^k_{ik,j}+(\Gamma^k_{kl}\Gamma^l_{ij}-\Gamma^k_{jl}\Gamma^l_{ik}) \right)\nonumber\\
&=& \sum\limits_{i,j,k,m=1}^3g^{ij}g^{km}(g_{ik, jm}-g_{ij,km})+Q(\partial g) \label{scexp},
\end{eqnarray} where $Q(\partial g)$ is a quadratic function in the partials $\partial g$ contracted with the metric $g$.   We can apply (\ref{scexp}) to the metric $g=\tilde g^{\theta}+h^{\theta}=g_{\text{Eucl}}+\tilde{h}^{\theta}+h^{\theta}$ to obtain the following: there is a $C_1>0$ so that for small $\epsilon$, we have 
\begin{eqnarray} 
R(\tilde{g}^{\theta}+h^{\theta})&=& \sum\limits_{i,j=1}^3 \big(\tilde{g}^{\theta}_{ij,ij}-\tilde g^{\theta}_{ii,jj}\big)+ \sum\limits_{i,j=1}^3 \big(h^{\theta}_{ij,ij}-h^{\theta}_{ii,jj}\big)+B(\tilde h^{\theta}, h^{\theta})+Q(\tilde h^{\theta})\nonumber\\
&=& \sum\limits_{i,j=1}^3 \big(\tilde{g}^{\theta}_{ij,ij}-\tilde g^{\theta}_{ii,jj}\big)+ L(h^{\theta})+B(\tilde h^{\theta}, h^{\theta})+Q(\tilde h^{\theta}), \label{Scalin}
\end{eqnarray}
 where 
\begin{eqnarray*}
|B(\tilde h^{\theta}, h^{\theta})|&\leq &C_1 \|\tilde h^{\theta}\|_{C^2(\Omega)} \cdot \|h^{\theta}\|_{C^2(\Omega)}\leq C_1C^2 \epsilon^2,\\  
|Q(\tilde h^{\theta})|&\leq &C_1 \|\tilde h^{\theta}\|_{C^2(\Omega)}^2\leq C_1 C^2 \epsilon^2.
\end{eqnarray*}

We also note that for a Schwarzschild metric $g^S$ of mass $m$ and center $c$, i.e. $g^S(x)=\Big(1+\frac{m}{2|x-c|}\Big)^4 g_{\text{Eucl}}$, the leading order term in the expansion $\sum\limits_{i,j=1}^3\left( g^S_{ij,i}-g^S_{ii,j}\right)\nu_e^j$ (where $\nu_e=\frac{x}{|x|}$ is the Euclidean unit normal to the sphere $|x|=r$) is just $$-4 \left(1+\frac{m}{2|x|}\right)^3\left(\frac{m}{2|x|^2}\right)\sum\limits_{i,j=1}^3\left( \frac{x^i}{|x|}\delta_{ij}-\frac{x^j}{|x|}\right)\frac{x^j}{|x|}.$$  Let $d\mu_e$ be the Euclidean surface measure.  Then there are functions $F_\ell(m, c, r)$, with $|F_\ell(m, c, r)|\leq C_2 m^2/r$, where $C_2$ is independent of $|c|\leq r/2$,  so that 
\begin{eqnarray}
\label{eq:schmass}
 \int\limits_{|x|=r} \sum\limits_{i,j=1}^3\left( g^S_{ij,i}-g^S_{ii,j}\right) \nu_e^j\; d\mu_e =16 \pi m +F_0(m,c,r),
\end{eqnarray}
and for $\ell=1, 2, 3$, 
\begin{eqnarray}
 \int\limits_{|x|=r}\left[ 
\sum\limits_{i,j=1}^3x^\ell \left( g^S_{ij,i}-g^S_{ii,j}\right) \nu_e^j d\mu_e - 
\sum\limits_{i=1}^3\big(g^S_{i\ell} \nu_e^i- g^S_{ii}\nu_e^\ell \big)  d\mu_e\right] = 16\pi mc^\ell+F_\ell(m,c,r).\nonumber\\
\label{eq:schctr}
\end{eqnarray} 

Using integration by parts along with (\ref{Scalin}) and (\ref{eq:schmass}), and the fact that $h^{\theta}$ and its partial derivatives vanish at $\partial \Omega$, we obtain
\begin{eqnarray*}
\int_{\Omega}R(\tilde{g}^{\theta}+h^{\theta})\; dx =16 \pi \epsilon (\mu-1)m_T + O(\epsilon^2).
\end{eqnarray*}

Let $\Sigma_k=\{ x: |x-c_k|= 2\}$ and $\Sigma_0=\partial B_0=\partial\{ |x|\leq R_0\}$, and let $c_0=0$.  We now use (\ref{Scalin}), integration by parts along with $L^*(x^\ell)=0$, as well as (\ref{eq:schmass}) and (\ref{eq:schctr}), to obtain for $\ell=1, 2, 3$, 
\begin{eqnarray*}
\int_{\Omega} x^\ell R(\tilde{g}^{\theta}+h^{\theta})\; dx &=& \int_{\Omega} x^\ell \sum\limits_{i,j=1}^3 \big(\tilde{g}^{\theta}_{ij,ij}-\tilde g^{\theta}_{ii,jj}\big)\; dx +O(\epsilon^2) \\
&=&\sum\limits_{k=0}^N \int\limits_{\Sigma_k} \Big(x^\ell \sum\limits_{i,j=1}^3 \big(\tilde{g}^{\theta}_{ij,i}-\tilde g^{\theta}_{ii,j}\big) \nu_e^j  - \sum\limits_{i=1}^3( \tilde{g}^{\theta}_{i\ell}\nu_e^i-\tilde{g}^{\theta}_{ii}\nu_e^\ell) \Big)d\mu_e +O(\epsilon^2)\\
&=& \sum\limits_{k=0}^N \int\limits_{\Sigma_k} \Big((x^\ell-c_k^\ell) \sum\limits_{i,j=1}^3 \big(\tilde{g}^{\theta}_{ij,i}-\tilde g^{\theta}_{ii,j}\big) \nu_e^j  - \sum\limits_{i=1}^3( \tilde{g}^{\theta}_{i\ell}\nu_e^i-\tilde{g}^{\theta}_{ii}\nu_e^\ell) \Big)d\mu_e\\
& & +  16\pi\sum\limits_{k=1}^N \epsilon m_k c_k^\ell +O(\epsilon^2)\\
&=&\int\limits_{\Sigma_0} \Big(x^\ell \sum\limits_{i,j=1}^3 \big(\tilde{g}^{\theta}_{ij,i}-\tilde g^{\theta}_{ii,j}\big) \nu_e^j  - \sum\limits_{i=1}^3( \tilde{g}^{\theta}_{i\ell}\nu_e^i-\tilde{g}^{\theta}_{ii}\nu_e^\ell) \Big)d\mu_e +O(\epsilon^2)\\
&=& 16\pi \epsilon \mu m_T c^\ell + O(\epsilon^2),
\end{eqnarray*}
where we have used the condition $\sum\limits_{k=1}^N m_k c_k=0$.  

Thus we see by continuity that for $\epsilon$ small enough, for some $\mu$ and $c$, with $(\mu-1)=O(\epsilon)$ and $c= O(\epsilon)$, that we can arrange the above integrals to vanish.  This means for $\theta=(\epsilon, \mu, c)$, the resulting metric $\bar{g}_{\epsilon}:=\tilde g^{\theta}+h^{\theta}$ is scalar-flat and is exactly Schwarzschild (with the respective masses) near $\partial \Omega$, so we extend it to $\mathbb{R}^3\setminus B_0$ using the exterior Schwarzschild metric.  

To finish the proof of the theorem, we note that by applying Theorem \ref{thm:schw}, we can assume without loss of generality that for each $k$, there is an $R_k>1$ and there is a coordinate system $x$ for $E_k$, so that on $\{ |x|\geq R_k\}$, $g^k$ has the Schwarzschild form  $g^k(x)= \big( 1+\frac{m_k}{2|x|}\big)^4 \delta_{ij}$.  Indeed, in applying the construction in \cite{cor:schw} to the given metrics in Theorem \ref{main}, the original masses can be made to change by an arbitrarily small amount (in our case by less than $\frac{\epsilon}{2N}$, say), by taking $R_k=2R$ sufficiently large; we can assume as well that $|x|<R$ on each $U_k$.   We assume this has been done and the masses have been re-labelled to $m_1, \ldots, m_N$.  

Let $\Omega_{\alpha}=\mathbb{R}^3\setminus \Big( \bigcup\limits_{k=1}^N \{ x: |x-c_k/\alpha|\leq \frac{1}{\alpha}\}\Big)$.  We can now scale the metric $\bar g_{\epsilon}$ by pulling back a scaling of coordinates $f_{\epsilon}: \Omega_{\epsilon} \rightarrow \Omega_1$, $f_{\epsilon}(x)=\epsilon x$, for $\epsilon<\frac{1}{2R}$.  In the resulting metric $g_{\epsilon}=\frac{1}{\epsilon^2} f_{\epsilon}^*\bar g_{\epsilon}$, the masses of the Schwarzschild neighborhoods of the boundaries $|x-c_k/\epsilon|=\frac{1}{\epsilon}>2R$ are $m_1, m_2, \ldots, m_N$, and near the outer boundary is $\mu m_T$.  Given $\epsilon_0>0$, we can choose $\epsilon$ small enough so that $\big| m(g_{\epsilon})- m_T\big| = \big| (\mu-1) m_T\big|< \epsilon_0$.  Note that the resulting configuration of centers is the original configuration $c_1,\ldots, c_N$ scaled by $\frac{1}{\epsilon}$.  Finally we can glue in the metrics $g^k$ on $\big( E_k\cap \{x:|x|\leq \frac{2}{\epsilon}\} \big)\supset U_k$ isometrically, by identifying $\{ x: \frac{1}{\epsilon}\leq|x-c_k/\epsilon|\leq \frac{2}{\epsilon}\}\subset \Omega_{\epsilon}$ with $E_k\cap\{ x: \frac{1}{\epsilon}\leq |x| \leq \frac{2}{\epsilon}\}$.  \end{proof}

\begin{remark}  In the above construction of $\tilde{g}^{\theta}$ and $\bar{g}_{\epsilon}=\tilde g^{\theta}+h^{\theta}$, we could allow one or more of the $m_k$ to be non-positive.  The proof above uses the positive mass metrics from Proposition \ref{prop:as}, which were used to interpolate the various Schwarzschild metrics about each $c_k$, so this requires $m_T>0$.  One could instead interpolate using the Euclidean metric if one also adds another gluing annulus near the boundary $\partial B_0$ to patch the Euclidean metric to an exterior Schwarzschild, and the proof follows \emph{mutatis mutandis}.  Using this variant of the construction, we could in principle allow cases with $m_T< 0$. \end{remark}

\section{Apparent Horizons}  A natural question to ask is whether the $N$-body configurations we produce contain apparent horizons (minimal spheres), apart from any minimal spheres already present in the bodies.  We want to know whether we introduce any new horizons in the construction, which would shield one or more of the $N$ bodies from the others; we also note that the construction in Theorem \ref{thm:schw} is a small perturbation near infinity in an AE end, and so it does not introduce any new minimal surfaces into the geometry.  In this section we show that for $n=3$ and for small enough $\epsilon$, no minimal surfaces intersect $\Omega_{\epsilon}$.  This can be proved by various means (cf. \cite{cm}), and one such way is to use the following proposition from \cite{cor:mh}, which relies on the Riemannian Penrose Inequality \cite{br:pen, hi:pen}, cf. \cite{BrayLee}.  For completeness we record the proof here. 

\begin{prop}
Let $(M,g)$ be a complete AE three-manifold with nonnegative scalar curvature.  Suppose that the sectional curvatures are bounded above by $C>0$.  Then if the ADM mass $m$ satisfies $m\sqrt{C}<\frac{1}{2}$, there are no closed minimal surfaces in $M$. \label{masstop}
\end{prop}

\begin{proof} If there were a minimal surface in $M$, then by the analysis in \cite{hi:pen} (which uses the convex barrier spheres near infinity in an end), there would exist an outermost minimal sphere $\Sigma$.  Let $A(\Sigma)$ denote its area.

We note that the Gauss equation implies an upper bound on the curvature of the horizon $\Sigma$ as follows.  Let $p$ be a point in $\Sigma$, and let $\{e_{1}, e_{2}\}$ be a basis of $T_{p}(\Sigma)$ in which the second fundamental form II is diagonal.  Let $\kappa_{1}=$II$(e_{1},e_{1}), \kappa_{2}=$II$(e_{2},e_{2})$ be the principal curvatures of $\Sigma$ at $p$.  Since $\Sigma$ is minimal, we have $\kappa_{1}+\kappa_{2}=0$, and so $\kappa_{1}\kappa_{2}\leq 0$.  The Gauss equation then yields
\[ K^{\Sigma} \equiv R(e_{1},e_{2},e_{1},e_{2}) = \overline{R}(e_{1},e_{2},e_{1},e_{2}) +  \kappa_{1}\kappa_{2} \leq C \]
\noindent where $R$ denotes the curvature tensor of $\Sigma$, and $\overline{R}$ the curvature tensor of $M$.  

We now invoke the Gauss-Bonnet theorem, which together with the preceding inequality yields 
\[ 4\pi = \int_{\Sigma} K^{\Sigma} \, d\mu_{\Sigma} \leq C\; A(\Sigma). \]
Combining this with the Penrose Inequality $m \geq \sqrt{\frac{A(\Sigma)}{16\pi}}$, we obtain $$m\sqrt{C}\geq \frac{1}{2}.$$  \end{proof} 

We now consider complete AE solutions $(M_k, g^k)$, $n\geq 3$, $k=1, \ldots, N$, of the time-symmetric constraint $R(g^k)=0$, in each of which we have chosen an AE end $E_k$ and compactly contained sub-domains $U_k\subset E_k$.  We apply the construction in the proof of Theorem \ref{main} to obtain a family of AE metrics $(M,g_{\epsilon})$; we can write, using above notation, $M= \Omega'_{\epsilon} \cup \overline{V}_{\epsilon}$, where 
\begin{eqnarray*}
\overline{V}_{\epsilon}&=& \bigcup \limits_{k=1}^N \Big( (M_k\setminus E_k)\cup \{x\in E_k: |x|\leq R\} \cup T_k \Big)\; ,\\
\Omega'_{\epsilon}&=&\mathbb{R}^3\setminus \Big( \bigcup\limits_{k=1}^N \{ x: |x-c_k/\epsilon|\leq \frac{2}{\epsilon}\}\Big)=M_{\text{ext}}\setminus \overline V_{\epsilon}\; .
\end{eqnarray*}  
Recall $U_k\subset \{x\in E_k: |x|\leq R\}$, and we have let $T_k$ be the \emph{transition region}  $\{x\in E_k: R\leq |x|\leq \frac{2}{\epsilon})$, comprised of the annulus where the original metric transitions to the appropriate Schwarzschild metric (as in Theorem \ref{thm:schw}), together with another annulus on which the metric is Schwarzschild.  We note that $g_{\epsilon}$ converges as $\epsilon \rightarrow 0^+$ to the Euclidean metric on $\Omega_{\epsilon}$, and each transition region is foliated by convex spheres. 

\begin{prop} For sufficiently small $\epsilon$, any closed minimal surface in $(M,g_{\epsilon})$ must be contained in $V_{\epsilon}\setminus \bigcup\limits_{k=1}^N T_k$.  \end{prop}

\begin{proof}   If there were a minimal surface whose intersection with $M\setminus V_{\epsilon}$ were nonempty, then applying arguments in \cite{hi:pen}, there is an outermost minimal sphere $\Sigma_{\epsilon}$ with respect to the end contained in $M_{\text{ext}}$ into which we have glued $V_{\epsilon}$, with $\Sigma\cap (M\setminus V_{\epsilon})$ nonempty.  Let $V_{\epsilon}'=V_{\epsilon} \setminus \left( \bigcup\limits_{k=1}^N \{ x\in E_k: \frac{1}{\epsilon}\leq |x-c_k|\leq \frac{2}{\epsilon}\}\right)$.  We re-scale the metric by a constant so that on $\Omega_{\epsilon}$ the metric is $f_{\epsilon}^*\bar{g}_{\epsilon}$, which is uniformly close to the Euclidean metric.  Now it follows by applying the argument in Proposition \ref{masstop} that for small $\epsilon$ there cannot be an outermost horizon contained in $\Omega_{\epsilon}$.   Thus $\Sigma_{\epsilon}$ will clearly have to intersect both $V_{\epsilon}'$ and $M\setminus V_{\epsilon}$.   By the Penrose inequality, and the estimate of $\mu$, there is a constant $C_2$ so that $A(\Sigma_{\epsilon})\leq C_2\epsilon^2$.  

On the other hand, let $p\in \Sigma_{\epsilon}\cap (M\setminus V_{\epsilon})$.  There is an $r_0\in (0,1)$ and for each $\epsilon$ an $r_{\epsilon}\in (r_0, 1)$ so that if $B_{r}^M(p)$ is the geodesic ball of radius $r$ about $p$, then $S=\Sigma_{\epsilon}\cap B_{r_{\epsilon}}^M(p)$ is a smooth surface with nonempty boundary (since $\Sigma_{\epsilon}\cap V_{\epsilon}'$ is nonempty).  We can now apply Schoen's curvature estimates for the stable hypersurface $S$ \cite{s:ems}: there is a $C$ so that for all small $\epsilon$, the second fundamental form on $\Sigma_{\epsilon}\cap B_{r_0/2}^M(p)$ has norm bounded by $C$.  Hence near $p$, $\Sigma_{\epsilon}\cap B_{r_0/2}^M(p)$ is graphical with bounded gradient and curvature over a ball of fixed size in $T_p \Sigma_{\epsilon}$.  The area of $\Sigma_{\epsilon}$ is thus uniformly bounded from below, either by the graphical description or volume comparison (since the intrinsic ball $B^{\Sigma_{\epsilon}}_{r_0/2}(p) \subset \Sigma_{\epsilon}\cap B_{r_0/2}^M(p)$).  This contradicts the area estimate from the Penrose Inequality. 

Thus $\Sigma_{\epsilon} \subset \overline{V}_{\epsilon}$.  But then by the maximum principle, it follows that $\Sigma_{\epsilon}\subset V_{\epsilon}$, since $\partial \Omega'_{\epsilon} \cap \overline{V}_{\epsilon}$ is a union of $N$ convex spheres.  Similarly, $\Sigma_{\epsilon}$ cannot intersect any of the transition regions. \end{proof}

\end{document}